\documentclass{cta-author}

\usepackage{graphicx}
\usepackage{amsmath, amssymb, amsthm, mathtools, stmaryrd, bm, gensymb} 
\usepackage{url}
\usepackage{tikz, pgfplots}
\usepackage{siunitx}
\usepackage{algorithm, algpseudocode}	
\usepackage{siunitx}

\mathtoolsset{showonlyrefs} 

\usetikzlibrary{plotmarks}

\newcommand{\mc}[1]{\ensuremath{\mathcal{#1}}}				
\newcommand{\bmm}[1]{\ensuremath{\mathbb{#1}}}				
\newcommand{\tran}{\mathsf{T}}						
\newcommand{\hermit}{\mathsf{H}}					
\newcommand{\frob}[1]{\ensuremath{\left\|#1\right\|_\textrm{F}}}	
\newcommand{\esp}[1]{\ensuremath{\mathbb{E}\left[#1\right]}}		

\newcommand{\Rss}{\bm{R}_{ss}}
\newcommand{\Rbb}{\bm{R}_{bb}}
\newcommand{\Rxx}{\bm{R}_{xx}}
\newcommand{\pxs}{\bm{p}_{xs}}
\newcommand{\edes}{\bm{e}_{d}}

\newtheorem{theorem}{Theorem}

\begin{document}

\supertitle{Research Article}

\title{Low-Complexity Separable Beamformers for Massive Antenna Array Systems}

\author{\au{Lucas N. Ribeiro$^{1\corr}$}, \au{Andr\'e L. F. de Almeida$^{1}$}, \au{Josef A. Nossek$^{1,2}$}, \au{Jo\~ao C\'esar M. Mota$^1$}}

\address{
	\add{1}{Wireless Telecommunications Research Group, Federal University of Cear\'a, Fortaleza, Brazil.}
	\add{2}{Department of Electrical and Computer Engineering, Technical University Munich, Munich, Germany.}
\email{nogueira@gtel.ufc.br}}

\begin{abstract}
	
	Future cellular systems will likely employ massive bi-dimensional arrays to improve performance by large array gain and more accurate spatial filtering, motivating the design of low-complexity signal processing methods. We propose optimising a Kronecker-separable beamforming filter that takes advantage of the bi-dimensional array geometry to reduce computational costs. The Kronecker factors are obtained using two strategies: alternating optimisation, and sub-array minimum mean square error beamforming with Tikhonov regularization. According to the simulation results, the proposed methods are computationally efficient but come with source recovery degradation, which becomes negligible when the sources are sufficiently separated in space.
	
\end{abstract}

\maketitle

\section{Introduction}

The number of wireless connected devices has been growing significantly, bringing new challenges to engineers. Future mobile communications systems, for example, are expected to provide very high throughput to several mobile terminals. In order to boost system capacity, new transceiver and network architectures are under investigation. Massive multiple-input-multiple-output (MIMO) technology, which consists of employing a large number of antenna elements at the base station (BS) to serve many multi-antenna users, is expected to yield significant spectral efficiency improvement \cite{larsson_massive_2014,schwarz_society_2016}. Such massive systems should be implemented using planar arrays in order to reduce the array's physical dimensions and to perform elevation and azimuth beamforming. This implementation, known as full-dimension MIMO (FD-MIMO), allows for better interference mitigation and has already been incorporated into 3GPP standards~\cite{ji_overview_2017}. These technologies pose new engineering challenges concerning computational and energy efficiency~\cite{ribeiro2018energy}, calling for research efforts to design computationally efficient signal processing methods for high-dimensional systems.~\footnote{\textcolor{red}{This paper is a preprint of a paper accepted by IET Signal Processing and is subject to Institution of Engineering and Technology Copyright. When the final version is published, the copy of record will be available at the IET Digital Library}} 

The high computational complexity of multidimensional filtering systems is not a new problem, though, and the first attempts to tackle this problem can be traced back to some decades ago. For instance, the authors in \cite{treitel_design_1971} proposed a multi-stage representation for bi-dimensional filters based on the coefficient matrix eigendecomposition, yielding computer storage and speed savings. However, computing the eigendecomposition of high-dimensional observations is expensive in general. More recent works have been interested in exploiting the algebraic structure present in some problems to reduce computational costs and to improve system performance. The authors in~\cite{bousse_tensor-based_2017} introduced tensor-based blind source separation methods which reduce the number of parameters to be estimated by exploiting the structure of low-rank signals. Such property implies that signals can be well approximated by a finite sum of low-dimensional Kronecker products. Although this representation simplifies the parameter estimation problem, signal accuracy is degraded. Kronecker separability has also been exploited in~\cite{rupp_tensor_2015,rupp2015gradient,pinheiro_nonlinear_2016} to increase the convergence rate of adaptive algorithms. Gradient descent-based solutions were presented in \cite{rupp_tensor_2015,rupp2015gradient} to identify second-order Kronecker separable systems, which can be useful to model telephone hybrid causing electrical echoes. In~\cite{pinheiro_nonlinear_2016}, the authors show that Volterra systems with separable kernels can be expressed in terms of Kronecker products. In~\cite{ribeiro2015identification}, we introduce a supervised system identification method to identify third-order Kronecker separable impulse responses based on alternating optimisation. The proposed identification method is applied to identify the telephone hybrid-like impulse response of~\cite{rupp_tensor_2015}. Simulation results indicate that the proposed method exhibits better accuracy than the classical Wiener filter solution. In~\cite{paleologu2018linear}, the method of~\cite{ribeiro2015identification} is extended to cope with low-rank Kronecker separable systems, allowing for the identification of more intricate acoustic responses. In \cite{elisei2018efficient}, fast recursive least squares methods for identifying second-order Kronecker separable (bilinear) systems are presented. Analytical and simulation results confirm the low computational costs and the identification performance of the proposed bilinear methods.

It is well-known that the spatial signature of planar arrays can be decomposed along its two dimensions~\cite{van_trees_optimum_2002}. Based on this property, beamforming techniques have been proposed. The authors in~\cite{wang_two-dimensional_2017_top} obtained a low-complexity two-dimensional MIMO precoding scheme by exploiting the Kronecker structure in the steering vectors of rectangular arrays. Therein, the proposed separable zero-forcing (ZF) precoder is presented based on the small angular spread at the elevation domain assumption, which enables algebraic separation of the azimuth and elevation domains by filtering. Results show that when this assumption is satisfied, the proposed separable ZF filter exhibits acceptable performance. However, in more realistic scenarios, this assumption is seldom met, and the performance of the separable ZF filter is severely degraded. In~\cite{zhu2017hybrid}, a clever hybrid analogue/digital beamforming method based on the Kronecker product is proposed for multi-cell multi-user MIMO systems. The analogue beamformers are designed exploiting the mixed product property of the Kronecker product to null inter-cell interference and to enhance the desired signal power. In \cite{miranda_generalized_2015}, the authors investigate the performance of a tensor global sidelobe canceller (GSC). Simulation results suggest that this tensor-based beamformer requires fewer snapshots than the classical GSC filter to achieve the desired performance. A tensor minimum variance distortionless response beamformer has been introduced in \cite{liu2018robust} for polarization sensitive arrays. In~\cite{ribeiro_low-complexity_nodate}, we express the received signal vector of a massive MIMO system equipped with a planar array using multilinear (tensor) algebra. From this model, we derived a two-step low-complexity equaliser which exploits each signal dimension, similar to~\cite{wang_two-dimensional_2017_top}. It basically consists of sub-array ZF beamforming followed by a low-dimensional minimum mean square error (MMSE) equaliser.

In this present work, we propose novel beamforming techniques which exploit the array separability to reduce their computational costs. Our methods are based on the classical MMSE beamformer, also known as Wiener filter~\cite{haykin_adaptive_1996}. This filter may be computationally impractical due to the inversion of a possibly very large covariance matrix. The matrix inversion lemma~\cite{petersen_matrix_2012} can be applied to the MMSE beamformer when the signal statistics are perfectly known, however, in practice, this seldom happens. As alternatives to the classical MMSE solution, we propose methods which aim at optimising a beamforming filter with Kronecker structure, i.e., the coefficients vector admits a Kronecker factorization. Thus, instead of optimising a large beamforming vector, we propose designing two relatively small beamforming vectors corresponding to the Kronecker factors. We present two strategies to design a Kronecker separable filter. In the first strategy, the mean square error (MSE) function is minimised by means of alternating optimisation. This strategy was first introduced in \cite{ribeiro_tensor_2016}, where the beamforming filter is obtained using sample estimates of the received signal covariance matrix. Here, we derive analytical expressions for the beamformer assuming perfect knowledge of the array manifold matrix. The second strategy consists of a closed-form solution based on the Khatri-Rao factorization of the separable array manifold matrix. Each sub-beamformer is obtained by performing sub-array MMSE beamforming with Tikhonov regularization. Simulation results show that the proposed methods can be computationally efficient, however, they come with source recovery degradation, which becomes insignificant when the wavefronts are sufficiently separated in the space.

The following notation is adopted throughout the paper: $x$ denotes a scalar, $\bm{x}$ a vector, and $\bm{X}$ a matrix. The $(i,j)$-th entry of $\bm{X}$ is given by $[\bm{X}]_{i,j}$. The transposed, conjugated transposed (Hermitian), and pseudo-inverse of $\bm{X}$ are denoted by $\bm{X}^\tran$, $\bm{X}^\hermit$, and $\bm{X}^\dagger$, respectively. The $(M\times M)$-dimensional identity matrix is represented by $\bm{I}_M$. The absolute value, the Frobenius and $\ell_2$ norms, and the expected value operator are respectively denoted by $|\cdot|$, $\frob{\cdot}$, $\| \cdot \|_2$, and $\esp{\cdot}$. The Kronecker, Khatri-Rao, and $n$-mode products are represented by $\otimes$, $\diamond$, and $\times_n$, respectively. $O(\cdot)$ represents the Big-O notation.

This work is organized as follows: the system model is introduced in~Section~\ref{sec:signal} and the proposed beamforming methods are presented in Section~\ref{sec:methods}. Therein, we also discuss their computational complexity. Simulation results are shown and discussed in Section~\ref{sec:simulations}, and the work is concluded in Section~\ref{sec:conc}.

\section{System  Model} \label{sec:signal}

Consider a multi-antenna system equipped with a uniform rectangular array (URA) with $N_h$ antennas in the horizontal axis, and $N_v$ in the vertical axis. This array of $N = N_h N_v$ antennas is distributed along the $y$-$z$ plane, as illustrated in Figure~\ref{fig:array}. Each antenna element has the same beam pattern $g(\phi, \theta)$, where $\phi$ and $\theta$ denote the azimuth and elevation angles, respectively\footnote{In practical antenna arrays, the element beampatterns would be different due to phenomena like mutual coupling, among others. To model such scenario, one would need to consider individual antenna beampatterns $g_n(\phi,\theta)$ for all $n \in \{1,\ldots,N\}$.}. The array is illuminated by $R$ independent narrow-band wavefronts in far-field propagation arriving from directions $(\phi_r, \theta_r),\,r=1,\ldots,R$ and carrying digitally-modulated signals. The wavefronts are assumed to have the same wavelength $\lambda$. The modulated signals at discrete-time instant $k$ are denoted by $s_r[k]$ and assumed to be mutually uncorrelated with zero mean and variance~$\sigma_{s}^2$.

The received signal at the $n$-th antenna can be modelled as the superposition of the $R$ incoming wavefronts:
\begin{equation} \label{eq:recv}
	x_n[k] = \sum_{r=1}^R g(\phi_r, \theta_r) a_n(\phi_r, \theta_r) s_r[k] + b_n[k],
\end{equation}
where $a_n(\phi_r, \theta_r)$ denotes the array response to the $r$-th wavefront at the $n$-th antenna, and $b_n[k]$ the complex additive white Gaussian sensor noise (AWGN) with zero mean and variance $\sigma_b^2$. The inter-antenna spacing in both the horizontal and vertical axes is $d_h=d_v=\lambda/2$, thus the array response can be written as
\begin{equation} \label{eq:antresp}
	a_n(\phi_r, \theta_r) = e^{j \pi[(n_h-1)\sin\phi_r\sin\theta_r+(n_v-1)\cos\theta_r]}.
\end{equation}
with $n = n_h + (n_v-1)N_h$, $n_h \in \{1,\ldots, N_h\}$, $n_v \in \{1,\ldots, N_v\}$. For notation simplicity, we define direction cosines with respect to the horizontal and vertical axis as $p_r = \sin\phi_r\sin\theta_r$ and $q_r = \cos\theta_r$, respectively. Then, using matrix notation and assuming omni-directional antennas, the received signals vector ${\bm{x}[k] = \left[ x_1[k], \ldots, x_N[k] \right]^\tran}$ can be represented as
	\begin{equation}
	\bm{x}[k] = \sum_{r=1}^R \bm{a}(p_r, q_r) s_r[k] + \bm{b}[k] = \bm{A} \bm{s}[k] + \bm{b}[k], \label{eq:recv_vec}
	\end{equation}
where $\bm{a}(p_r, q_r) = \left[ a_1(p_r, q_r), \ldots, a_N(p_r, q_r)\right]^\tran$ stands for the array steering vector, $\bm{s}[k] = \left[ s_1[k], \ldots, s_R[k]\right]^\tran$ the symbols vector, and $\bm{b}[k]=[b_1[k],\ldots, b_N[k]]^\tran$ the AWGN vector. Note that the model \eqref{eq:recv_vec} is valid only for a specific angular range where $g(\phi_r, \theta_r) = 1$. Now the array manifold matrix can be written as
\begin{equation} \label{eq:array_manifold}
\bm{A} = [ \bm{a}(p_1, q_1), \ldots, \bm{a}(p_R, q_R) ] \in \bmm{C}^{N \times R}.
\end{equation}
From our assumptions, it follows that the covariance matrix of the received signals is given by
\begin{equation}
\Rxx = \esp{\bm{x}[k]\bm{x}[k]^\hermit} = \bm{A} \Rss \bm{A}^\hermit + \Rbb,
\end{equation}
where $\Rss=\esp{\bm{s}[k]\bm{s}[k]^\hermit} = \sigma_s^2 \bm{I}_R$ and $\Rbb = \esp{\bm{b}[k] \bm{b}[k]^\hermit} = \sigma_b^2 \bm{I}_{N}$. The multi-antenna system employs a beamformer to recover a desired signal among the $R$ incoming signals. We define the signal-to-noise ratio (SNR) as the desired signal power over the AWGN variance, i.e. $\text{SNR} = \sigma_s^2 / \sigma_b^2$.

\begin{figure}[t]
	\centering
	\includegraphics[width=0.7\linewidth]{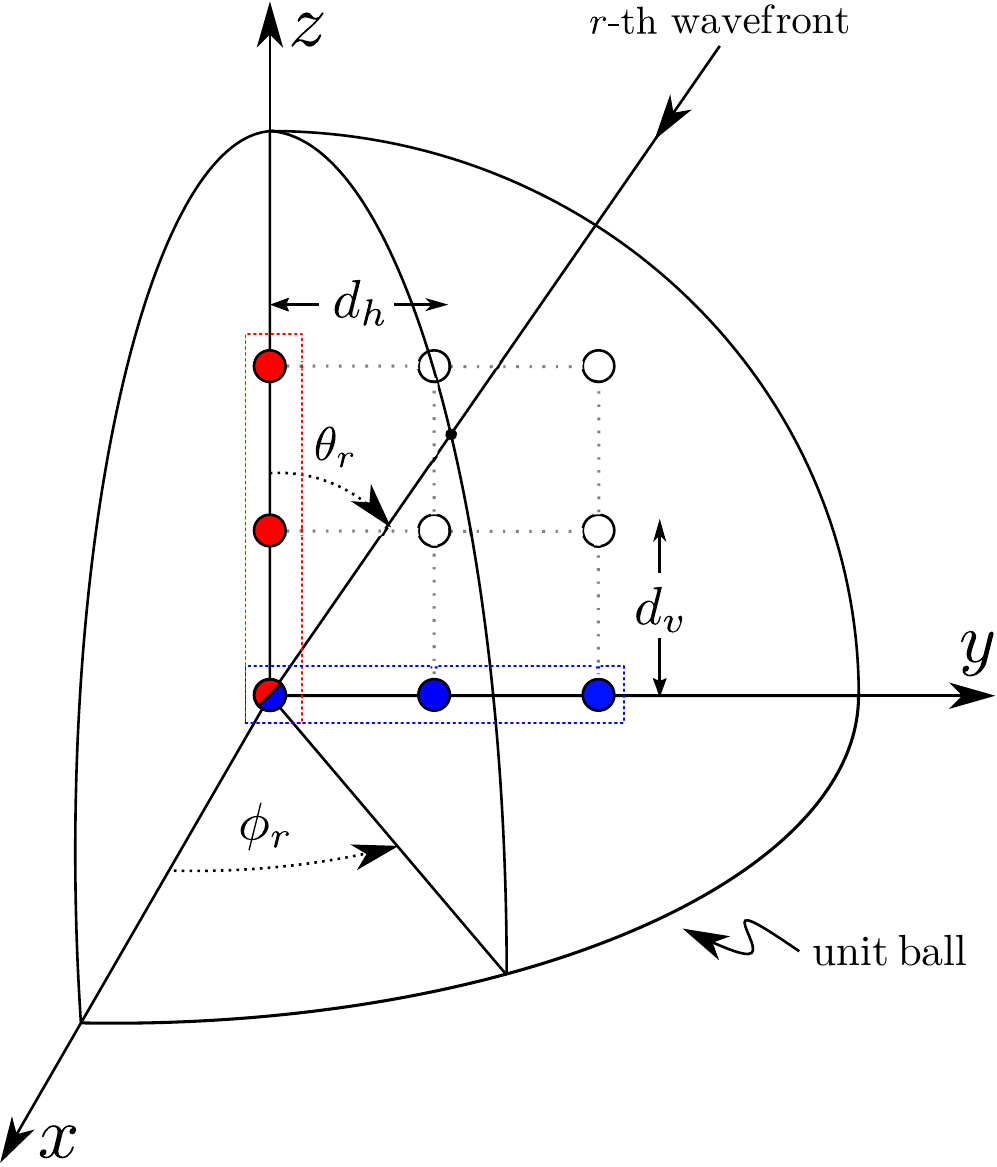}
	\caption{Uniform Rectangular Array (URA) in the $y$-$z$ plane.}
	\label{fig:array}
\end{figure}

The array response can be \emph{separated} into horizontal and vertical contributions owing to the URA bi-dimensionality~\cite{van_trees_optimum_2002}. More specifically, the array response with respect to any wavefront can be factorized as 
\begin{equation} \label{eq:sepan}
a_n(p_r, q_r) = a_{n_h}(p_r) a_{n_v}(q_r),
\end{equation}
where $a_{n_h}(p_r) = e^{j \pi (n_h-1) p_r}$ and $a_{n_v}(q_r) = e^{j \pi (n_v-1) q_r}$. The sub-array steering vectors are then defined as
\begin{align}
\bm{a}_h(p_r) &= \left[ a_1(p_r), \ldots, a_{N_h}(p_r) \right]^\tran,\\
\bm{a}_v(q_r) &= \left[ a_1(q_r), \ldots, a_{N_v}(q_r) \right]^\tran.
\end{align}
The horizontal and vertical sub-arrays of a URA are depicted in Figure~\ref{fig:array}. The separable representation in \eqref{eq:sepan} leads to the Kronecker factorization of the array steering vectors: 
\begin{equation}
\bm{a}(p_r, q_r) = \bm{a}_v(q_r) \otimes \bm{a}_h(p_r),
\end{equation}
and, consequently, the array manifold matrix \eqref{eq:array_manifold} can be written as a Khatri-Rao product
	\begin{align}
		\bm{A} &= \left[ \bm{a}_v(q_1) \otimes \bm{a}_h(p_1), \ldots, \bm{a}_v(q_R) \otimes \bm{a}_h(p_R) \right]= \bm{A}_v \diamond \bm{A}_h, \hspace{0.25cm} \label{eq:sep_manifold_matrix}
	\end{align}
where 
\begin{align}
\bm{A}_h &= \left[ \bm{a}_h(p_1),\ldots, \bm{a}_h(p_R)\right] \in \bmm{C}^{N_h \times R},\\
\bm{A}_v &= \left[ \bm{a}_v(q_1),\ldots,\bm{a}_v(q_R) \right] \in \bmm{C}^{N_v \times R},
\end{align}
stand for the vertical and horizontal sub-array manifold matrices, respectively. Equation~\eqref{eq:sep_manifold_matrix} emphasizes the separable structure of the URA and shall be exploited in beamforming design.

\section{Beamforming Methods} 
\label{sec:methods}

We are interested in spatially filtering the received signals $\bm{x}[k]$ to extract $s_d[k]$, the signal of $d$-th (desired) wavefront, while attenuating the interfering signals. To this end, we design the beamforming filter $\bm{w} \in \bmm{C}^{N}$ so that its output  $y[k]=\bm{w}^\hermit \bm{x}[k]$ approximates the desired signal.
We choose to optimise this filter to minimise the mean square error (MSE) function
\begin{align} \label{eq:msefunc}
J_\text{MSE}(\bm{w}) &= \bmm{E}\left[|s_d[k] - \bm{w}^\hermit \bm{x}[k]|^2 \right] \\
			         &= \sigma_s^2 - \pxs^\hermit \bm{w} - \bm{w}^\hermit \pxs + \bm{w}^\hermit \Rxx \bm{w},
\end{align}
where $\pxs = \bmm{E}\left[ \bm{x}[k] s_d^*[k] \right] = \bm{A} \Rss \bm{e}_{d} \in \bmm{C}^{N}$ denotes the cross-covariance vector, and $\bm{e}_{r} \in \bmm{C}^{R}$ the $r$-th canonical vector in the $R$-dimensional space. The MMSE beamformer yields the global minimum of $J_\textrm{MSE}(\bm{w})$ and is given by the Wiener filter $\bm{w}_{opt} = \Rxx^{-1} \pxs$~\cite{haykin_adaptive_1996}. For large array systems, the computation of this filter becomes impractical since it involves the inversion of a very large covariance matrix. Iterative algorithms, such as the gradient descent method, can be used to simplify the calculations, however each of their iterations can still be computationally expensive. 

To simplify the calculations of the MMSE beamforming filter, we impose the following Kronecker structure: $\bm{w} = \bm{w}_v \otimes \bm{w}_h$, $\bm{w}_m \in \bmm{C}^{N_m}$, $m \in \{v,h\}$. Such a representation is motivated by the computational reduction of the beamformer design, since only $(N_v + N_h)$ parameters need to be optimised, against $N_v N_h$ when separability is not considered. In order to gain more insight into the array separability, let us consider an example with $N$ antennas and $R=1$ impinging wavefront. The received signal in this case is given by $\bm{x}[k] = \bm{a}(p_d,q_d)s[k] + \bm{b}[k]$. The output signal for the filter $\bm{w}$ is then written as
\begin{equation}
	y[k] = \bm{w}^\hermit \bm{x}[k] = \text{AF} \cdot  s[k] + \bm{w}^\hermit \bm{b}[k],
\end{equation} 
where $\text{AF} = \bm{w}^\hermit \bm{a}(p_d,q_d)$ is the array factor. Note that it can be rewritten as
\begin{equation}
	\text{AF} = \left[ \bm{w}_v^\hermit \bm{a}_v(q_d) \right] \cdot \left[ \bm{w}_h^\hermit \bm{a}_h(p_d) \right]. \label{eq:af}
\end{equation}
Equation~\eqref{eq:af} shows that the total array factor is given by the product of the sub-array factors. Note that this property does not depend on the beam pattern of the antenna elements, since it only relies on the factorization of the array factor. The steering vectors of some array geometries, such as circular arrays, for example, do not permit a Kronecker factorization. In this case, we cannot directly apply the methods proposed in this work.

We present two novel beamforming strategies based on the MMSE filter that exploits array separability to reduce computational costs. In the first strategy, we recast the MSE function \eqref{eq:msefunc} using tensor algebra, and then we devise an iterative beamformer based on alternating minimisation. The reader is referred to \cite{kolda_tensor_2009,comon2009tensor} for an introduction to tensor algebra. In the second strategy, we obtain a closed-form beamforming filter by employing sub-array MMSE filtering. In the end, we discuss the computational complexity of the proposed methods.

\subsection{Tensor MMSE Beamformer}
\label{sec:iterative_analyitical}

Let us first reformulate the received signal model \eqref{eq:recv} using tensor algebra. Considering array separability \eqref{eq:sepan}, the received signal at the $n$-th antenna can be rewritten as
\begin{equation} \label{eq:mult_recv_sig}
	x_{n_h, n_v}[k] = \sum_{r=1}^{R} = a_{n_v}^{(v)}(q_r) a_{n_h}^{(h)}(p_r) s_r[k] + b_{n_h, n_v}[k],
\end{equation}
Now, define the received signals matrix $[\bm{X}[k]]_{n_h, n_v} = x_{n_h, n_v}[k]$, the array manifold tensor $[\mc{A}]_{n_h, n_v, r} =  a_{n_h}^{(h)}(p_r)a_{n_v}^{(v)}(q_r)$, and the AWGN matrix $[\bm{B}[k]]_{n_h, n_v} = b_{n_h, n_v}[k]$. Using tensor modal products \cite{kolda_tensor_2009}, the received signals matrix can be expressed as
\begin{equation} \label{eq:tensor_model}
	\bm{X}[k] = \mc{A} \times_3 \bm{s}[k]^\tran + \bm{B}[k] \in \bmm{C}^{N_h \times N_v}.
\end{equation}
The array manifold tensor $\mc{A}$ is a three-dimensional array with dimensions $N_h \times N_v \times R$. The two first array modes refer to the physical array dimensions, whereas the third one represents the transmitted signal dimension, i.e. the number of wavefronts. This tensor can be unfolded into matrices in three different manners \cite{kolda_tensor_2009}:
\begin{align}
	[\mc{A}]_{(1)} &= \left[ \bm{a}_h(p_1)\bm{a}_v(q_1)^\tran, \ldots, \bm{a}_h(p_R)\bm{a}_v(q_R)^\tran \right] \in \bmm{C}^{N_h \times N_v R},\\
	[\mc{A}]_{(2)} &= \left[ \bm{a}_v(q_1)\bm{a}_h(p_1)^\tran, \ldots, \bm{a}_v(q_R)\bm{a}_h(p_R)^\tran \right]\in \bmm{C}^{N_v \times N_h R},\\
	[\mc{A}]_{(3)} &= (\bm{A}_v \diamond \bm{A}_h)^\tran = \bm{A}^\tran \in \bmm{C}^{R \times N_v N_h}.
\end{align}

Let us now rewrite the beamformer output $y[k] = \bm{w}^\hermit \bm{x}[k]$ in terms of $\bm{w}_h$ and $\bm{w}_v$ by considering the Kronecker factorization of $\bm{w}$ and the bi-dimensional representation of the received signals \eqref{eq:mult_recv_sig}:
\begin{align}
y[k] &= \sum_{n = 1}^N [\bm{w}]_n^* x_n[k]\\
	 &= \sum_{n_h=1}^{N_h} \sum_{n_v=1}^{N_v} [\bm{w}_h]_{n_h}^* [\bm{w}_v]_{n_v}^*  x_{n_h, n_v}[k]. \label{eq:it_bf_out_nmode_a}
\end{align}
Using matrix notation, \eqref{eq:it_bf_out_nmode_a} can be rewritten as
\begin{equation}
y[k] = \bm{w}_h^\hermit \bm{X}[k] \bm{w}_v^* = \bm{w}_v^\hermit \bm{X}[k]^\tran \bm{w}_h^*. \label{eq:it_bf_out_nmode}
\end{equation}
The MSE function \eqref{eq:msefunc} can now be reformulated as the following bi-linear function
\begin{align}
	&J_\textrm{MSE}(\bm{w}_h, \bm{w}_v) = \label{eq:bi_msefunc}\\
	& \bmm{E}\left[ \left| s_d[k] - \bm{w}_h^\hermit \bm{X}[k] \bm{w}_v^*\right|^2 \right] =\bmm{E}\left[ \left| s_d[k] - \bm{w}_v^\hermit \bm{X}[k]^\tran \bm{w}_h^*\right|^2 \right].
\end{align}

Unfortunately, minimising \eqref{eq:bi_msefunc} is not straightforward. The gradient of $J_\textrm{MSE}(\bm{w}_h, \bm{w}_v)$ with respect to any of its vector variables depends on the other variable. This coupling disables the direct application of methods such as gradient descent, calling for alternating minimisation techniques. To this end, let us define the horizontal and vertical sub-array input signals
\begin{align}
	\bm{u}_h[k] &=  \bm{X}[k] \bm{w}_v^* \in \bmm{C}^{N_h}, \label{eq:it_subsig_h}\\
	\bm{u}_v[k] &=  \bm{X}[k]^\tran \bm{w}_h^* \in \bmm{C}^{N_v}.  \label{eq:it_subsig_v}
\end{align}
and rewrite \eqref{eq:bi_msefunc} as
\begin{align}
	J_\textrm{MSE}(\bm{w}_h, \bm{w}_v) &= \bmm{E}\left[ \left| s_d[k] - \bm{w}_h^\hermit \bm{u}_h[k]\right|^2 \right] \label{eq:it_out_lin1}\\
	&= \bmm{E}\left[ \left| s_d[k] - \bm{w}_v^\hermit \bm{u}_v[k]\right|^2 \right]. \label{eq:it_out_lin2}
\end{align}
It is easy to recognize \eqref{eq:it_out_lin1} and \eqref{eq:it_out_lin2} as linear functions of $\bm{w}_h$ and $\bm{w}_v$, respectively, when the other vector variable is fixed. The proposed beamforming method, referred to as Tensor MMSE (TMMSE), consists of sequentially minimising \eqref{eq:it_out_lin1} and \eqref{eq:it_out_lin2} using the MMSE solution for each sub-filter until a convergence criterion is satisfied. The sub-beamformers are calculated according to the following theorem:

\begin{theorem} \label{theorem:urra}
	The minimisers of \eqref{eq:it_out_lin1} and \eqref{eq:it_out_lin2} conditioned on $\bm{w}_v$ and $\bm{w}_h$ are respectively given by 
	\begin{align}
		\bm{w}_h &= \bm{R}_{hh}^{-1} \bm{p}_{hs}, \label{eq:mmse_h}\\
		\bm{w}_v &= \bm{R}_{vv}^{-1} \bm{p}_{vs},\label{eq:mmse_v}
	\end{align}
	where 
	\begin{align}
		\bm{R}_{hh} &= \bmm{E}\left[ \bm{u}_h[k] \bm{u}_h[k]^\hermit  \right] \\
							&= [\mc{A}]_{(1)} ( \Rss \otimes \bm{w}_v^*\bm{w}_v^\tran) [\mc{A}]_{(1)}^\hermit +\sigma_b^2 \|\bm{w}_v\|^2_2 \bm{I}_{N_h} \in \bmm{C}^{N_h \times N_h}, \label{eq:covh}\\
		\bm{R}_{vv}  &= \bmm{E}\left[ \bm{u}_v[k] \bm{u}_v[k]^\hermit  \right] \\
							&= [\mc{A}]_{(2)} ( \Rss \otimes \bm{w}_h^*\bm{w}_h^\tran) [\mc{A}]_{(2)}^\hermit + \sigma_b^2 \|\bm{w}_h\|^2_2 \bm{I}_{N_v} \in \bmm{C}^{N_v \times N_v} \label{eq:covv}
	\end{align}
	denote the covariance matrices of the sub-array input signals, and
	\begin{align}
		\bm{p}_{hs} &= \bmm{E}\left[ \bm{u}_h[k] s_d^*[k] \right] = [\mc{A}]_{(1)} ( \Rss\edes \otimes \bm{w}_v^*) \in \bmm{C}^{N_h}, \label{eq:crosscorrh}\\
		\bm{p}_{vs} &= \bmm{E}\left[ \bm{u}_v[k] s_d^*[k] \right] = [\mc{A}]_{(2)} ( \Rss\edes \otimes \bm{w}_h^*) \in \bmm{C}^{N_v} \label{eq:crosscorrv}
	\end{align}
	the cross-covariance vectors between the sub-array input signals and the signal of interest.	
\end{theorem} 
\begin{proof} \renewcommand{\qedsymbol}{}
	See the appendix.
\end{proof}

Theorem~\ref{theorem:urra} can be applied when the signals' statistics ($\Rss$ and $\Rbb$), and the array manifold matrix are known. However, such information might not be available in practice, and thus the sub-array covariance matrices and cross-covariance vectors need to be estimated. It can be done by  using sample estimates over $K$ time snapshots. In this sense,
the covariance matrices $\bm{R}_{hh}$ and $\bm{R}_{vv}$ can be estimated as
\begin{align}
\hat{\bm{R}}_{hh} &= \frac{1}{K} \sum_{k=0}^{K-1} \bm{u}_h[k] \bm{u}_h[k]^\hermit, \label{eq:covh_sample}\\
\hat{\bm{R}}_{vv} &= \frac{1}{K} \sum_{k=0}^{K-1} \bm{u}_v[k] \bm{u}_v[k]^\hermit, \label{eq:crosscorrh_sample}
\end{align}
and the cross-covariance vectors $\bm{p}_{hs}$ and $\bm{p}_{vs}$ as
\begin{align}
\hat{\bm{p}}_{hs} &= \frac{1}{K} \sum_{k=0}^{K-1} \bm{u}_h[k] s_d^*[k], \label{eq:covv_sample}\\
\hat{\bm{p}}_{vs} &=  \frac{1}{K} \sum_{k=0}^{K-1} \bm{u}_v[k] s_d^*[k].  \label{eq:crosscorrv_sample}
\end{align}
Note that $\bm{u}_h[k]$ and $\bm{u}_v[k]$ can be easily formed by observing $\bm{x}[k]$, reshaping into $\bm{X}[k]$, and using Equations \eqref{eq:it_subsig_h} and \eqref{eq:it_subsig_v}, respectively. The steps to compute the TMMSE beamformer are summarized in Algorithm~\ref{alg:tmmse_alg}.

\begin{algorithm}[t]
	\caption{Tensor MMSE algorithm}
	\label{alg:tmmse_alg}
	\begin{algorithmic}[1]
		\State Randomly initialize $\bm{w}_h$ and $\bm{w}_v$
		\Repeat
		\State Form $\bm{R}_{hh}$ and $\bm{p}_{hs}$
		\State $\bm{w}_h \gets \bm{R}_{hh}^{-1} \bm{p}_{hs}$
		\State Form $\bm{R}_{vv}$  and $\bm{p}_{vs}$
		\State $\bm{w}_v \gets \bm{R}_{vv}^{-1} \bm{p}_{vs}$
		\Until{convergence criterion triggers}
		\State $\bm{w} \gets \bm{w}_v \otimes \bm{w}_h$
	\end{algorithmic}
\end{algorithm}

\subsection{Kronecker MMSE Beamformer} \label{sec:kmmse}

Let us consider the following Khatri-Rao product property. Let $\bm{A} \in \bmm{C}^{P \times M}$, $\bm{B} \in \bmm{C}^{Q\times N}$, $\bm{C} \in \bmm{C}^{M \times R}$, and $\bm{D} \in \bmm{C}^{N \times R}$. From \cite{liu_hadamard_2008}, it follows that
\begin{equation} \label{eq:krkron}
	(\bm{A} \otimes \bm{B}) (\bm{C} \diamond \bm{D}) = (\bm{AC}) \diamond (\bm{BD}) \in \bmm{C}^{PQ \times R}.
\end{equation}
This result suggests that a Kronecker separable beamformer can be individually applied to the corresponding sub-array manifold matrix in \eqref{eq:sep_manifold_matrix}. In this case, the filtering operation $y[k] = \bm{w}^\hermit \bm{x}[k]$ can be carried out as
\begin{align}
y[k] &= (\bm{w}_v \otimes \bm{w}_h)^\hermit (\bm{A}_v \diamond \bm{A}_h)\bm{s}[k] + \bm{w}^\hermit \bm{b}[k]\\
	    &= \left[(\bm{w}_v^\hermit \bm{A}_v) \diamond (\bm{w}_h^\hermit \bm{A}_h)\right]\bm{s}[k] + \bm{w}^\hermit \bm{b}[k].
\end{align}
Therefore, instead of optimising an $N$-dimensional beamformer for $\bm{A}$, we can design two independent low-dimensional beamformers for $\bm{A}_h$ and $\bm{A}_v$ individually. According to this approach, each sub-beamformer is fed only with signals from the corresponding antenna sub-array. In this sense, we define the horizontal and vertical observed signals:
\begin{align}
\bm{x}_h[k] &= \bm{A}_h \bm{s}[k] + \bm{b}_h[k] \in \bmm{C}^{N_h} \\
\bm{x}_v[k] &= \bm{A}_v \bm{s}[k] + \bm{b}_v[k] \in \bmm{C}^{N_v},
\end{align}
where $\bm{b}_h[k] \in \bmm{C}^{N_h}$ and $\bm{b}_v[k] \in \bmm{C}^{N_v}$ represent the additive Gaussian noise vector observed at the horizontal and vertical sub-arrays, respectively. These vectors are defined as
\begin{align}
[\bm{b}_h[k]]_{n_h} &= \left. b_{n_h + (n_v-1)N_h} [k]\right|_{n_v=1},\\
[\bm{b}_v[k]]_{n_v} &= \left. b_{n_h + (n_v-1)N_h} [k]\right|_{n_h=1}.
\end{align}

We propose to optimise each sub-beamformer according to the MMSE criterion. However, the direct application of the MMSE filter to each sub-beamformer would be prone to numerical problems. Often in many practical scenarios, e.g. mobile communications, different signals are closely separated in an angular domain (azimuth or elevation). In this case, either the vertical or horizontal sub-array manifold matrices become almost rank deficient, turning the MSE minimisation problem ill-posed. To overcome this issue, we resort to Tikhonov regularization~\cite{palomar2010convex}, which avoids singular covariance matrices by penalizing large-norm solutions. The proposed beamforming method, hereafter referred to as Kronecker MMSE (KMMSE), independently minimises the following cost functions
\begin{align}
	J_\text{MSE}^{(h)}(\bm{w}_h, \rho) &= \bmm{E}\left[|s_d[k] - \bm{w}_h^\hermit \bm{x}_h[k]|^2 \right] + \rho \| \bm{w}_h \|_2^2, \label{eq:km1}\\
	J_\text{MSE}^{(v)}(\bm{w}_v, \rho) &= \bmm{E}\left[|s_d[k] - \bm{w}_v^\hermit \bm{x}_v[k]|^2 \right] + \rho \| \bm{w}_v \|_2^2, \label{eq:km2}
\end{align} 
where $\rho \geq 0$ denotes the regularization parameter. Define
\begin{align}
\bm{R}_m &= \bm{A}_m \Rss \bm{A}_m^\hermit + \bm{R}_{bb,m}, \label{eq:subcov}\\
\bm{p}_m &= \bm{A}_m \Rss \edes, \label{eq:subcrossc}
\end{align}
with $\bm{R}_{bb,m} = \sigma_b^2 \bm{I}_{N_m}$ for $m \in \{h,v\}$. The minimisers for \eqref{eq:km1} and \eqref{eq:km2} are thus given by $\bm{w}_m = (\bm{R}_m + \rho \bm{I}_{N_m})^{-1}\bm{p}_m$ for $m \in \{h,v\}$. Due to regularization, the KMMSE output signal is not guaranteed to have the same power as the desired signal. Thus, we employ the following scaling to correct the KMMSE output power: $y_{\text{KMMSE}}[k] = (\sigma_s/ \sigma_p) p[k] $, where $p[k] = (\bm{w}_v \otimes \bm{w}_h)^\hermit \bm{x}[k]$ and $\sigma_p$ denotes the standard deviation of $p[k]$. In a practical implementation, this scaling correction can be performed by the automatic gain control circuit. The computation of the KMMSE filter is summarized in Algorithm~\ref{alg:kmmse}.

In practice, one might not have \emph{a priori} knowledge of the sub-array manifold matrices ($\bm{A}_h$ and $\bm{A}_v$) and signals' statistics. One can estimate \eqref{eq:subcov} and \eqref{eq:subcrossc} using the received signals from the horizontal and vertical sub-arrays, represented by
\begin{align}
\bar{\bm{x}}_m[k] &= \bm{A}_m \bm{s}[k] + \bm{b}_m[k],\,m \in \{h, v\}. \label{eq:xh}
\end{align}
For the horizontal sub-array, we define 
\[[\bar{\bm{x}}_h[k]]_{n_h} = \left. x_{n_h + (n_v-1)N_h} [k]\right|_{n_v=1} = x_{n_h}[k],\] 
with $n_h \in \{1,\ldots,N_h\}$ and $r \in \{1,\ldots,R\}$. Similarly, for the vertical sub-array: 
\[[\bar{\bm{x}}_v[k]]_{n_v} = \left. x_{{n_h + (n_v-1)N_h}} [k]\right|_{n_h=1} = x_{1+(n_v-1)N_h}[k],\] 
with $n_v \in \{1,\ldots, N_v\}$ and $r \in \{1,\ldots,R\}$. Now, the covariance matrices can be estimated as
\begin{align}
\hat{\bm{R}}_h &=\left( \frac{1}{K} \sum_{k=0}^{K-1} \bar{\bm{x}}_h[k] \bar{\bm{x}}_h[k]^\hermit \right),\\
\hat{\bm{R}}_v &= \left( \frac{1}{K} \sum_{k=0}^{K-1} \bar{\bm{x}}_v[k] \bar{\bm{x}}_v[k]^\hermit \right),
\end{align}
and the cross-covariance vectors as
\begin{align}
\hat{\bm{p}}_{h} &= \left( \frac{1}{K} \sum_{k=0}^{K-1} \bar{\bm{x}}_h[k] s_d^*[k] \right),\\
\hat{\bm{p}}_{v} &= \left( \frac{1}{K} \sum_{k=0}^{K-1} \bar{\bm{x}}_v[k] s_d^*[k] \right).
\end{align}

The proposed closed-form KMMSE beamformer can be seen as a sub-optimal solution which relies on a covariance matrix approximation. According to the mixed product property of the Kronecker product~\cite{liu_hadamard_2008}, the KMMSE beamformer can be expressed as
\begin{equation} \label{eq:krfilt}
\bm{w} = \left[ (\bm{R}_v + \rho \bm{I}_{N_v}) \otimes (\bm{R}_h + \rho \bm{I}_{N_h}) \right]^{-1} (\bm{p}_{v} \otimes \bm{p}_{h}).
\end{equation}
The Kronecker product of covariance matrices in \eqref{eq:krfilt} can be regarded as an approximation of $\Rxx$. Also, it is straightforward to see in \eqref{eq:krfilt} that the cross-covariance vector $\bm{p}_{xs}$ can be exactly factorized into $\bm{p}_{v} \otimes \bm{p}_{h}$. We now conduct an asymptotic analysis of KMMSE to provide insights on its performance.

First, consider the classical MMSE filter
\begin{equation}
\bm{w}_{opt} = \Rxx^{-1} \pxs = \left( \bm{A} \Rss \bm{A}^\hermit + \Rbb \right)^{-1} \bm{A} \Rss \edes.
\end{equation}
Applying the matrix inversion lemma \cite{petersen_matrix_2012}, its Hermitian vector can be written as
\begin{equation}
\bm{w}_{opt}^\hermit = \edes^\tran \left( \Rss^{-1} + \bm{A}^\hermit \Rbb^{-1} \bm{A} \right)^{-1} \bm{A}^\hermit \Rbb^{-1}.
\end{equation}
From the signal statistics assumptions in Section~\ref{sec:signal}, we have
\begin{equation} \label{eq:mmseinv}
\bm{w}_{opt}^\hermit = \edes^\tran \left( \frac{\sigma_b^2}{\sigma_s^2} \bm{I}_R + \bm{A}^\hermit \bm{A} \right)^{-1} \bm{A}^\hermit.
\end{equation}
Now we rewrite the Kronecker factors of the KMMSE filter using \eqref{eq:mmseinv} and for $\rho=0$ to obtain
\begin{align}
&\bm{w}^\hermit = \left[ \edes^\tran \left( \frac{\sigma_b^2}{\sigma_s^2} \bm{I}_R + \bm{A}_v^\hermit \bm{A}_v \right)^{-1} \bm{A}_v^\hermit \right] \otimes \\
&\left[ \edes^\tran \left( \frac{\sigma_b^2}{\sigma_s^2} \bm{I}_R + \bm{A}_h^\hermit \bm{A}_h \right)^{-1} \bm{A}_h^\hermit\right]. \label{eq:kmmse_asyn}
\end{align}

At high SNR, the noise power drops and $\sigma_b^2 \rightarrow 0$. If the inverse matrix $(\bm{A}_m^\hermit \bm{A}_m)^{-1}$ exists for $m \in \{h,v\}$, then 
\begin{equation} \label{eq:wkr_high}
	\bm{w}^\hermit \rightarrow \left( \edes^\tran \bm{A}_v^\dagger  \right) \otimes  \left( \edes^\tran \bm{A}_h^\dagger  \right).
\end{equation}
As expected, each sub-array beamformer converges to a ZF filter. Using \eqref{eq:krkron} and \eqref{eq:wkr_high}, we see that the KMMSE output signal at high SNR converges to
\begin{equation}
y[k] \rightarrow \left[ \left( \edes^\tran \bm{A}_v^\dagger \bm{A}_v \right) \diamond \left( \edes^\tran \bm{A}_h^\dagger \bm{A}_h \right) \right] \bm{s}[k] = s_d[k].
\end{equation}
The inverse $(\bm{A}_m^\hermit \bm{A}_m)^{-1}$ exists if and only if $\bm{A}_m^\hermit \bm{A}_m$ is not rank deficient, i.e., the wavefronts arrive from different directions. However, when the wavefronts are closely spaced in the angular domain, $\bm{A}_m^\hermit \bm{A}_m$ becomes ill-conditioned and the ZF filter performs poorly. Fortunately, when $\rho >0$, the inverse matrix is defined, allowing for desired signal recovery.

At low SNR the term $\frac{\sigma_b^2}{\sigma_s^2} \bm{I}_R $ dominates and \eqref{eq:kmmse_asyn} goes to
\begin{align}
	\bm{w}^\hermit &= \left( \edes^\tran \frac{\sigma_s^2}{\sigma_b^2}  \bm{A}_v^\hermit\right) \otimes \left( \edes^\tran \frac{\sigma_s^2}{\sigma_b^2}  \bm{A}_h^\hermit\right)\\
						 &= \frac{\sigma_s^4}{\sigma_b^4} \bm{a}_v(\phi_d,\theta_d)^\hermit \otimes \bm{a}_h(\phi_d,\theta_d)^\hermit. \label{eq:wkr_low}
\end{align}
Equation~\eqref{eq:wkr_low} shows that, as in the classical MMSE filter, the factors of $\bm{w}$ converge to matched filters which maximize the desired signal power. In this case, the KMMSE output signal can be written as ${y[k] \rightarrow \frac{\sigma_s^4}{\sigma_b^4} [ ( \bm{a}_v(\phi_d,\theta_d)^\hermit \bm{A}_v ) \diamond ( \bm{a}_h(\phi_d,\theta_d)^\hermit\bm{A}_h ) ] \bm{s}[k] + \bm{w}^\hermit\bm{b}[k]}$. If the incoming wavefronts are sufficiently separated in the angular domain, i.e., if all $\bm{a}_v$ and all $\bm{a}_h$ are mutually orthogonal, then ${y[k] \rightarrow \frac{\sigma_s^4}{\sigma_b^4} s_d[k] + \bm{w}^\hermit\bm{b}[k]}$. The analysis above shows that the proposed KMMSE filter is able to recover the desired signal from the received signals despite the covariance matrix approximations. Note that it is based on the assumption that the incoming signals are sufficiently separated in the angular domain. 

\begin{algorithm}[t]
	\caption{Kronecker MMSE filter}
	\label{alg:kmmse}
	\begin{algorithmic}[1]
		\State Select $\rho \geq 0$ 
		\State Form $\bm{R}_h$ and $\bm{p}_h$
		\State $\bm{w}_h \gets (\bm{R}_h + \rho \bm{I}_{N_h})^{-1} \bm{p}_h$
		\State Form $\bm{R}_v$ and $\bm{p}_v$
		\State $\bm{w}_v \gets (\bm{R}_v + \rho \bm{I}_{N_v})^{-1} \bm{p}_v$
		\State $\bm{w} \gets \bm{w}_v \otimes \bm{w}_h$
	\end{algorithmic}
\end{algorithm}

The proposed beamforming methods work with low-dimensional sub-array manifold matrices to decrease their computational complexity. However, this also reduces their degrees of freedom, which are important for attenuating interfering signals. An MMSE filter designed for $N$ antennas has $N$ degrees of freedom, i.e., it is capable of recovering the desired signal and rejecting $N-1$ undesired sources. Our separable beamforming framework, by contrast, offers $N_h$ and $N_v$ degrees of freedom for the horizontal and vertical sub-arrays. Therefore, the separable filter performance is limited by the least degree of freedom. Hence, the proposed methods are capable of recovering the desired signal and rejecting $\min(N_h, N_v)-1$ undesired sources. The proposed separable beamforming framework exchanges degrees of freedom for computational complexity reduction.

\subsection{Computational Complexity} \label{sec:comp}

The MMSE filter is known to be computationally complex. However, if the array manifold matrix $\bm{A}$ and the signal statistics $\Rss$ and $\Rbb$ are known, one can employ the matrix inversion lemma to the MMSE filter and obtain the low-complexity MMSE expression \eqref{eq:mmseinv}, in which a $R\times R$ matrix is inverted. However, this information may not be available, and then sample estimates are needed to compute the MMSE filter. In this case, the inversion lemma cannot be applied, and, thus an $N\times N$ covariance matrix is inverted in order to get the MMSE filter coefficients. Such an operation has complexity $O(N^3)$, which can be overwhelming for massive array systems. In this case, the proposed methods can be used since they are much less expensive in computational terms, as we show in the following.

The TMMSE filter calculates its beamformer coefficients through an iterative process of $I$ iterations, in which $N_h$- and $N_v$-dimensional matrices are inverted. Therefore, the TMMSE filter requires $O(I(N_h^3 + N_v^3))$ operations. Therefore, this method is less complex than the classical approach provided that $I$, $N_h$ and $N_v$ are not too large. The authors in \cite{yener_interference_2001} discussed the convergence of alternating MMSE-based methods and concluded that they are monotonically convergent. Other numerical properties such as convergence rate and stability are not discussed and, to the best of our knowledge, the investigation of these aspects remains a research challenge. The analytical convergence study of the proposed method is beyond the scope of this work.

The KMMSE filter is much simpler than the previous methods since it performs sub-array beamforming using closed-form solutions. To obtain the beamformer coefficients, one needs to invert $N_h$- and $N_v$-dimensional matrices only once. Thus, this method carries out $O(N_h^3 + N_v^3)$ operations.

\section{Simulation Results} \label{sec:simulations}

We present numerical results from simulations conducted to assess the performance of the proposed methods. At each simulation, signal data is generated as follows: $R$ independent sequences of $K$ QPSK-modulated symbols are generated to form $\bm{s}[k]$, for all $k \in \{1,\ldots,K\}$. Next, the direction cosines ($p_r$ and $q_r$) of the $R$ wavefronts are randomly generated according to a uniform distribution in the range $[-0.9,\,0.9]$ so the array manifold matrix $\bm{A}$ is formed. Note that selecting the direction cosines within this range ensures the omnidirectional propagation assumption of Section~\ref{sec:signal}. Finally, the observed signals \eqref{eq:recv_vec} are formed by contaminating the received symbols with additive noise.

We investigate the computational complexity and source recovery performance of the proposed beamforming methods in terms of floating point operations (flops) and uncoded bit error ratio (BER) of the desired signal. We choose BER as figure of merit because it reveals the noise and interference rejection performance. Therefore, if the beamforming operation is correctly carried out, then the interfering wavefronts are attenuated, and the desired signal BER decreases. The graphs in Figures~\ref{fig:ber_lambda}, \ref{fig:cond_lambda} and \ref{fig:ber_random} were obtained by averaging the results from $10^5$ independent experiments considering $R=4$ wavefronts, $N_h \times N_v = 8 \times 8$ antennas, and $K=1000$ symbols. Figures \ref{fig:bp_mmse}--\ref{fig:bp_kmmse_q}, however, were collected from a single experiment with $R=6$ wavefronts, $N_h \times N_v = 4 \times 4$ antennas, and $K=1000$ symbols. The parameter selection for the latter figures will be motivated in the following paragraphs. The convergence of the TMMSE method is achieved when the normalised filter residual between consecutive iterations is smaller than a tolerance value $\epsilon >0$, i.e. $\|\frac{\bm{w}_i}{\|\bm{w}_i\|_2} - \frac{\bm{w}_{i-1}}{\|\bm{w}_{i-1}\|_2}\|_2 < \epsilon$, where $i$ denotes the iteration number. In all experiments, we set $\epsilon=10^{-3}$. Preliminary simulations have shown that the average number of TMMSE iterations is~$5$.

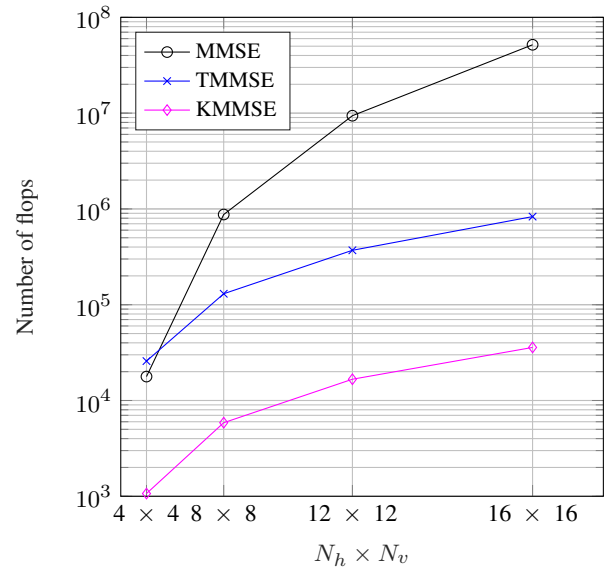
\begin{figure}[tb]
	\centering
%
%
\definecolor{mycolor1}{rgb}{1.00000,0.00000,1.00000}%
\begin{tikzpicture}

\begin{axis}[%
width=2.5in,
height=2.5in,
scale only axis,
xmin=0,
xmax=300,
xtick={16,64,144,256},
xticklabels={{$\text{4 }\times\text{ 4}$},{$\text{8 }\times\text{ 8}$},{$\text{12 }\times\text{ 12}$},{$\text{16 }\times\text{ 16}$}},
xlabel style={font=\color{white!15!black}},
xlabel={$N_h \times N_v$},
ymode=log,
ymin=1000,
ymax=100000000,
yminorticks=true,
ylabel style={font=\color{white!15!black}},
ylabel={Number of flops},
axis background/.style={fill=white},
xmajorgrids,
ymajorgrids,
yminorgrids,
legend style={font=\footnotesize,legend cell align=left, align=left, draw=white!15!black},
legend pos=north west
]
\addplot [color=black, mark=o, mark options={solid, black}]
  table[row sep=crcr]{%
16	17776\\
64	874456\\
144	9403696\\
256	51740536\\
};
\addlegendentry{MMSE}

\addplot [color=blue, mark=x, mark options={solid, blue}]
  table[row sep=crcr]{%
16	25789.932\\
64	130182.762\\
144	369978.884\\
256	831498.59\\
};
\addlegendentry{TMMSE}

\addplot [color=mycolor1, mark=diamond, mark options={solid, mycolor1}]
  table[row sep=crcr]{%
16	1068\\
64	5864\\
144	16676\\
256	35808\\
};
\addlegendentry{KMMSE}

\end{axis}
\end{tikzpicture}%
	\caption{Number of flops as function of array size for $R=4$ wavefronts.}	
	\label{fig:flops}
\end{figure}
~
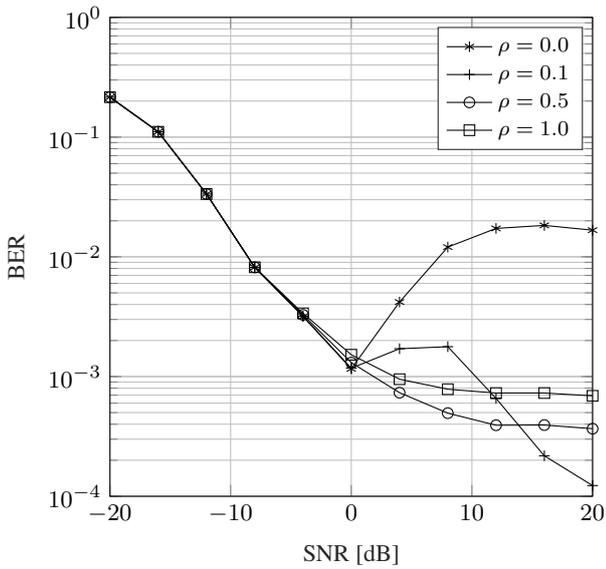
\begin{figure}[tb]
	\centering
%
%
\begin{tikzpicture}

\begin{axis}[%
width=2.5in,
height=2.5in,
scale only axis,
xmin=-20,
xmax=20,
xlabel style={font=\color{white!15!black}},
xlabel={SNR [dB]},
ymode=log,
ymin=0.0001,
ymax=1,
yminorticks=true,
ylabel style={font=\color{white!15!black}},
ylabel={BER},
axis background/.style={fill=white},
xmajorgrids,
ymajorgrids,
yminorgrids,
legend style={font=\footnotesize,legend cell align=left, align=left, draw=white!15!black}
]
\addplot [color=black, mark=asterisk, mark options={solid, black}]
  table[row sep=crcr]{%
-20	0.21543905\\
-16	0.11073812\\
-12	0.033475125\\
-8	0.008152135\\
-4	0.00317467\\
0	0.00115964\\
4	0.004192445\\
8	0.012062825\\
12	0.01727139\\
16	0.018304355\\
20	0.016702705\\
};
\addlegendentry{$\rho=0.0$}

\addplot [color=black, mark=+, mark options={solid, black}]
  table[row sep=crcr]{%
-20	0.21543905\\
-16	0.11073813\\
-12	0.03347524\\
-8	0.00815526\\
-4	0.003195545\\
0	0.00117454\\
4	0.00170837\\
8	0.001775885\\
12	0.000657085\\
16	0.000217925\\
20	0.00012289\\
};
\addlegendentry{$\rho=0.1$}

\addplot [color=black, mark=o, mark options={solid, black}]
  table[row sep=crcr]{%
-20	0.215439045\\
-16	0.11073812\\
-12	0.03347585\\
-8	0.00816771\\
-4	0.00327753\\
0	0.001308185\\
4	0.000733335\\
8	0.00049479\\
12	0.000392185\\
16	0.000393575\\
20	0.000366925\\
};
\addlegendentry{$\rho=0.5$}

\addplot [color=black, mark=square, mark options={solid, black}]
  table[row sep=crcr]{%
-20	0.21543905\\
-16	0.11073814\\
-12	0.03347655\\
-8	0.00818325\\
-4	0.003376125\\
0	0.00152173\\
4	0.000950555\\
8	0.000784295\\
12	0.00072845\\
16	0.00072887\\
20	0.00069047\\
};
\addlegendentry{$\rho=1.0$}

\end{axis}
\end{tikzpicture}%
	\caption{KMMSE BER performance for different regularization parameter $\rho$. $N_h \times N_v = 8 \times 8$, $R=4$ wavefronts.}
	\label{fig:ber_lambda}	
\end{figure}
~
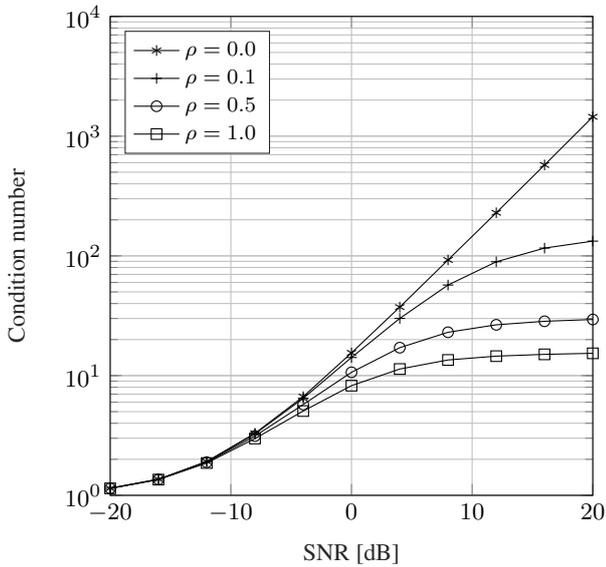
\begin{figure}[tb]
	\centering
%
%
\begin{tikzpicture}

\begin{axis}[%
width=2.5in,
height=2.5in,
scale only axis,
xmin=-20,
xmax=20,
xlabel style={font=\color{white!15!black}},
xlabel={SNR [dB]},
ymode=log,
ymin=1,
ymax=10000,
yminorticks=true,
ylabel style={font=\color{white!15!black}},
ylabel={Condition number},
axis background/.style={fill=white},
xmajorgrids,
ymajorgrids,
yminorgrids,
legend pos = north west,
legend style={font=\footnotesize, legend cell align=left, align=left, draw=white!15!black}
]
\addplot [color=black, mark=asterisk, mark options={solid, black}]
  table[row sep=crcr]{%
-20	1.14506517928489\\
-16	1.362400431239\\
-12	1.91627464354011\\
-8	3.2951954018796\\
-4	6.67580226057207\\
0	15.4826042711498\\
4	37.3414002507592\\
8	92.4573128574934\\
12	229.138476938298\\
16	574.377922896956\\
20	1450.7911206026\\
};
\addlegendentry{$\rho=0.0$}

\addplot [color=black, mark=+, mark options={solid, black}]
  table[row sep=crcr]{%
-20	1.14492025902587\\
-16	1.36149240337587\\
-12	1.91052959021697\\
-8	3.25938653844047\\
-4	6.45849563528776\\
0	14.1660038828634\\
4	30.0455004125021\\
8	57.0758459851947\\
12	89.2583766337681\\
16	116.111356911365\\
20	132.79919278204\\
};
\addlegendentry{$\rho=0.1$}

\addplot [color=black, mark=o, mark options={solid, black}]
  table[row sep=crcr]{%
-20	1.14434346197502\\
-16	1.35790534336025\\
-12	1.88825218165192\\
-8	3.12666829658553\\
-4	5.73356848289817\\
0	10.6550695140998\\
4	17.1091821801907\\
8	23.0125168794806\\
12	26.5632637181752\\
16	28.42731831831\\
20	29.4272768745575\\
};
\addlegendentry{$\rho=0.5$}

\addplot [color=black, mark=square, mark options={solid, black}]
  table[row sep=crcr]{%
-20	1.14362889038108\\
-16	1.35352040027196\\
-12	1.86189287930486\\
-8	2.98119686021205\\
-4	5.05963318128543\\
0	8.24130213557488\\
4	11.3481137444236\\
8	13.511990412044\\
12	14.5402337643168\\
16	15.0496899899918\\
20	15.3543675307171\\
};
\addlegendentry{$\rho=1.0$}

\end{axis}
\end{tikzpicture}%
	\caption{Condition number of \eqref{eq:subcov} for different regularization parameter $\rho$. $N_h \times N_v = 8 \times 8$, $R=4$ wavefronts.}
	\label{fig:cond_lambda}
\end{figure}
~
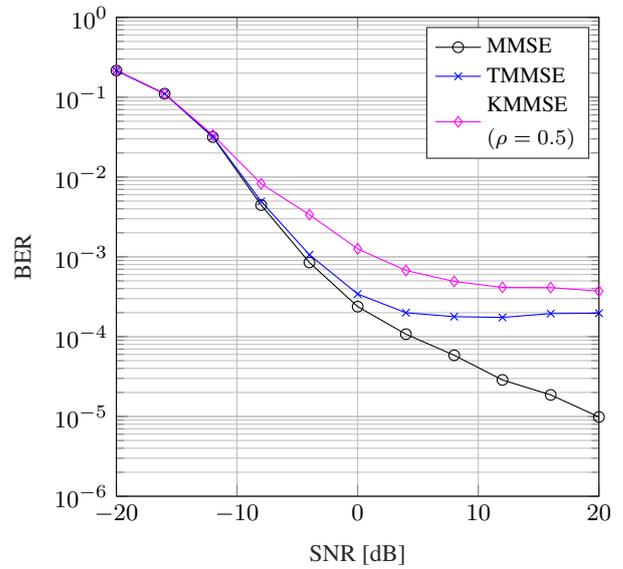
\begin{figure}[tb]
	\centering
%
%
\definecolor{mycolor1}{rgb}{1.00000,0.00000,1.00000}%
\begin{tikzpicture}

\begin{axis}[%
width=2.5in,
height=2.5in,
scale only axis,
xmin=-20,
xmax=20,
xlabel style={font=\color{white!15!black}},
xlabel={SNR [dB]},
ymode=log,
ymin=1e-06,
ymax=1,
yminorticks=true,
ylabel style={font=\color{white!15!black}},
ylabel={BER},
axis background/.style={fill=white},
xmajorgrids,
ymajorgrids,
yminorgrids,
legend style={font=\footnotesize,legend cell align=left, align=left, draw=white!15!black}
]
\addplot [color=black, mark=o, mark options={solid, black}]
  table[row sep=crcr]{%
-20	0.215326845\\
-16	0.11043396\\
-12	0.031633745\\
-8	0.004450225\\
-4	0.000849475\\
0	0.000237715\\
4	0.000107395\\
8	5.843e-05\\
12	2.868e-05\\
16	1.863e-05\\
20	9.86e-06\\
};
\addlegendentry{MMSE}

\addplot [color=blue,mark=x, mark options={solid, blue}]
  table[row sep=crcr]{%
-20	0.215331485\\
-16	0.11050138\\
-12	0.032035545\\
-8	0.00494125\\
-4	0.00106208\\
0	0.000342145\\
4	0.000199335\\
8	0.00017813\\
12	0.00017424\\
16	0.00019468\\
20	0.000196485\\
};
\addlegendentry{TMMSE}

\addplot [color=mycolor1, mark=diamond, mark options={solid, mycolor1}]
  table[row sep=crcr]{%
-20	0.215347585\\
-16	0.110712005\\
-12	0.03354346\\
-8	0.0082404\\
-4	0.003373465\\
0	0.0012597\\
4	0.000675255\\
8	0.00049212\\
12	0.00041396\\
16	0.000412285\\
20	0.00037213\\
};
\addlegendentry{KMMSE\\$(\rho=0.5)$}

\end{axis}
\end{tikzpicture}%
	\caption{$N_h \times N_v = 8 \times 8$, $R=4$ wavefronts.}
	\label{fig:ber_random}	
\end{figure}

\begin{figure}[tb]
	\centering
	\includegraphics[width=\linewidth]{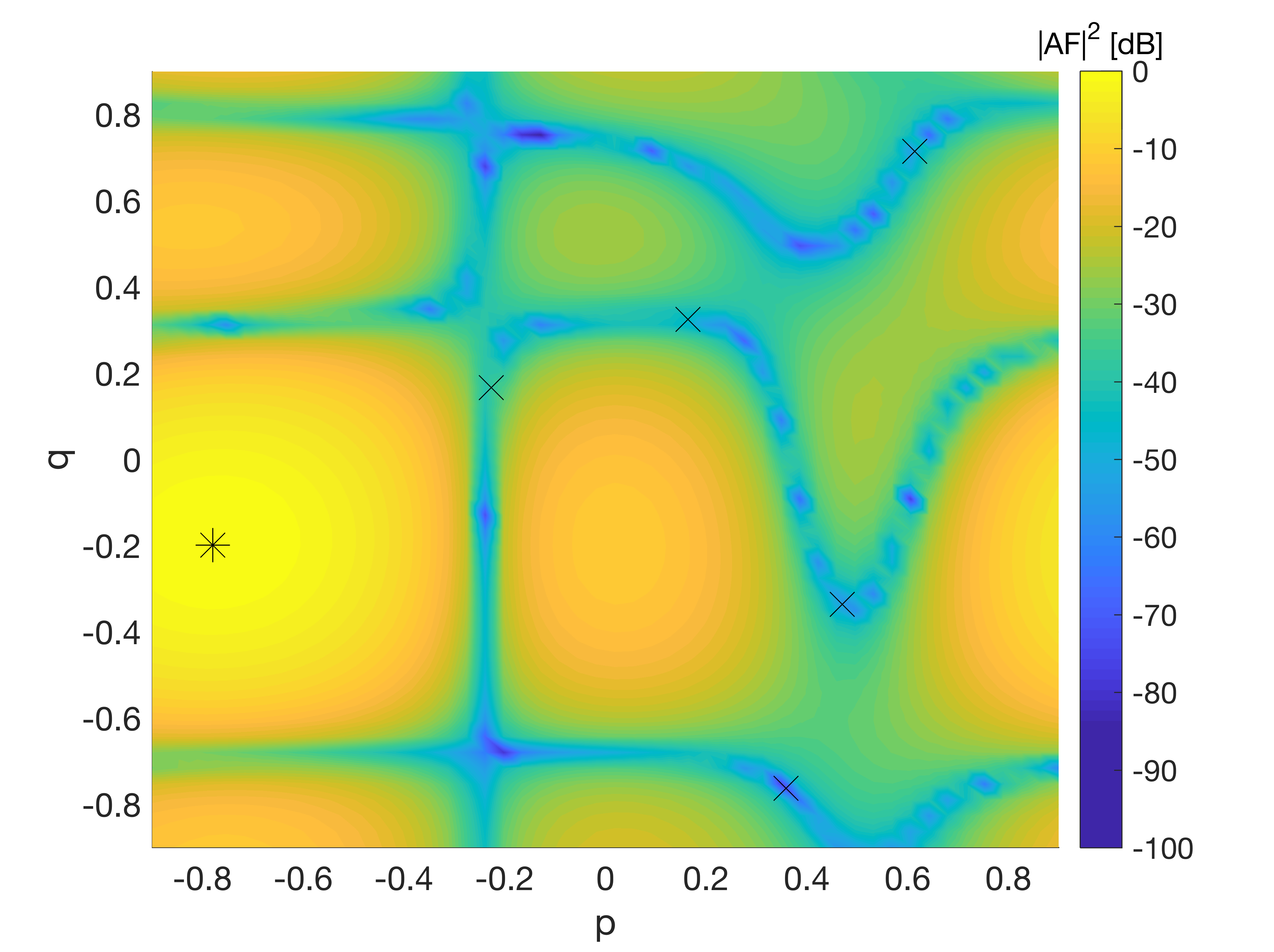}
	\caption{MMSE AF squared magnitude. $R=6$ wavefronts, $N_h \times N_v = 4 \times 4$. Asterisk denotes desired signal, cross interfering signal.}	
	\label{fig:bp_mmse}
\end{figure}
~
\begin{figure}[tb]
	\centering
	\includegraphics[width=\linewidth]{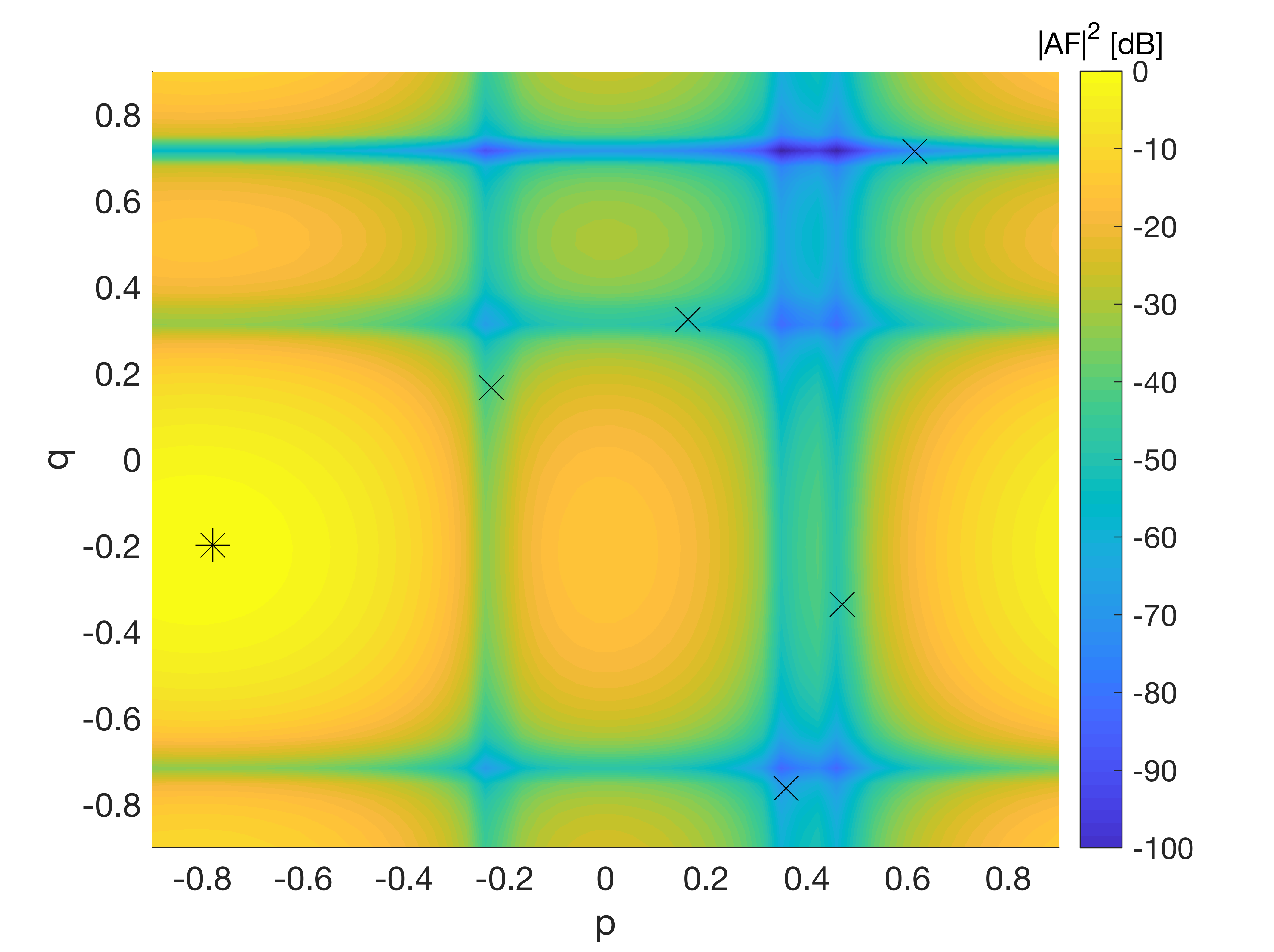}
	\caption{TMMSE AF squared magnitude. $R=6$ wavefronts, $N_h \times N_v = 4 \times 4$. Asterisk denotes desired signal, cross interfering signal.}	
	\label{fig:bp_tmmse}	
\end{figure}
~
\begin{figure}[tb]
	\centering
	\includegraphics[width=\linewidth]{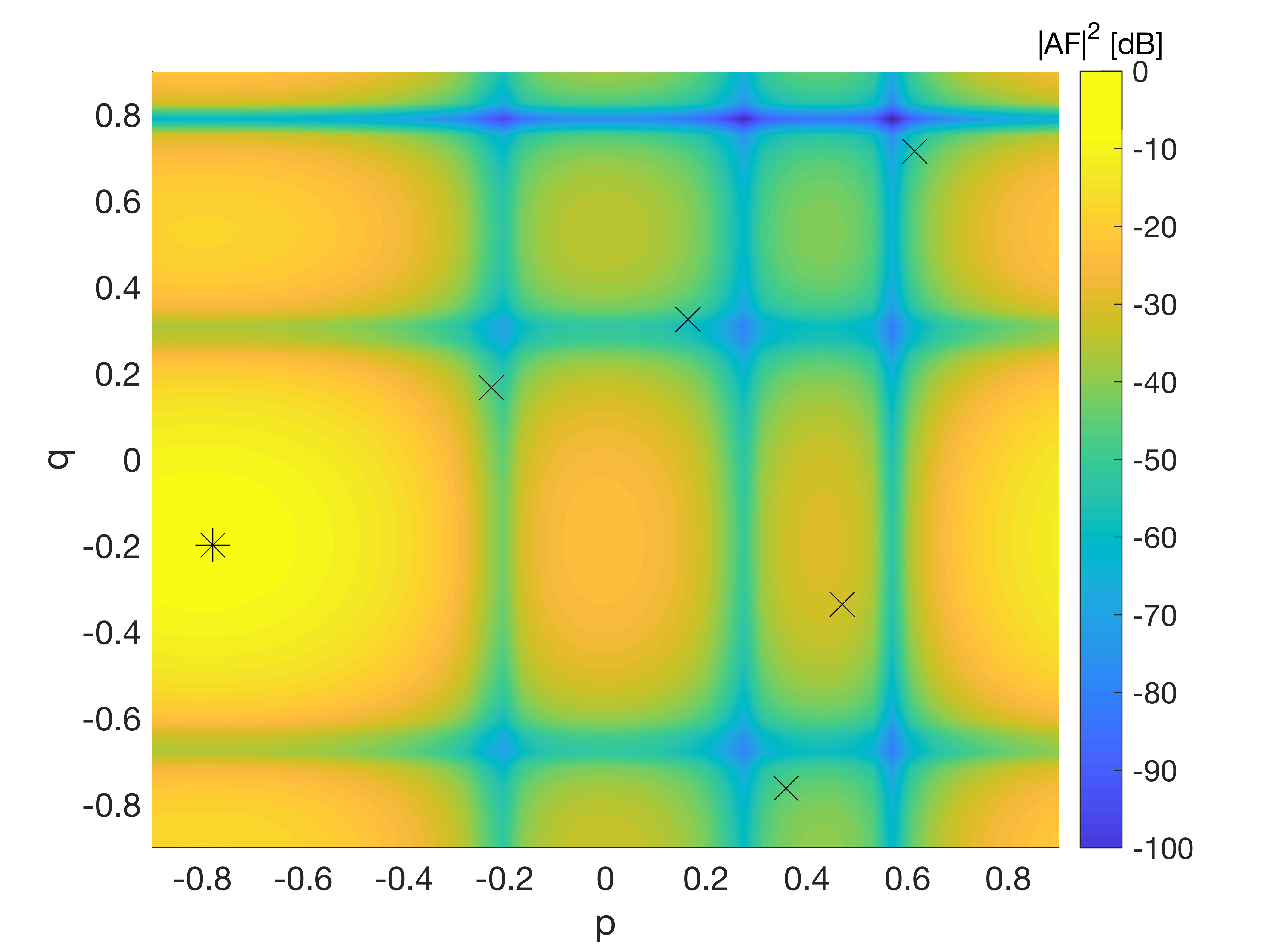}
	\caption{KMMSE AF squared magnitude. $R=6$ wavefronts, $N_h \times N_v = 4 \times 4$, $\rho=0.5$. Asterisk denotes desired signal, cross interfering signal.}	
	\label{fig:bp_kmmse}	
\end{figure}
~
\begin{figure}[tb]
	\centering
	\includegraphics[width=\linewidth]{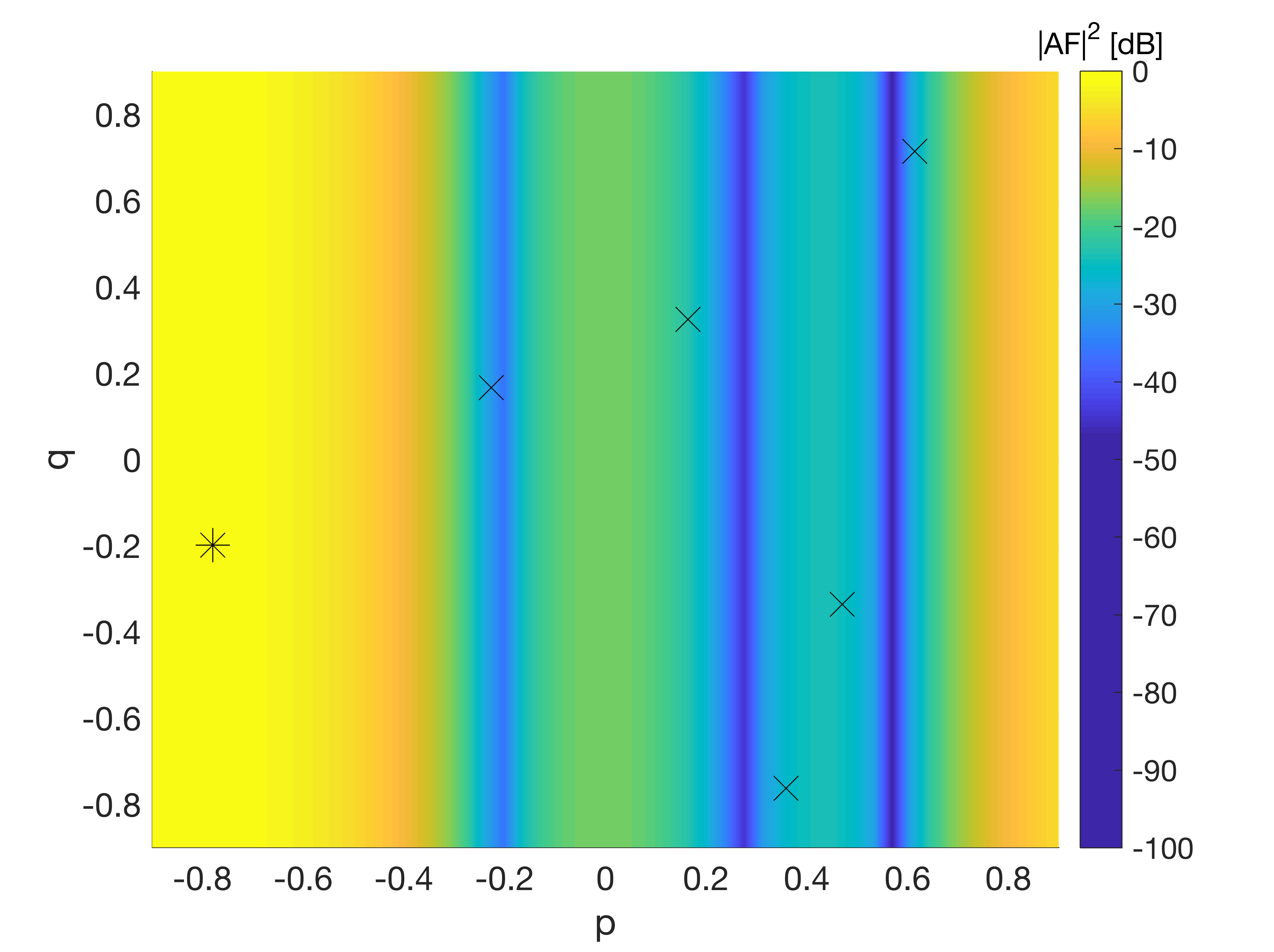}
	\caption{Horizontal KMMSE AF squared magnitude. $R=6$ wavefronts, $N_h= 4$, $\rho=0.5$. Asterisk denotes desired signal, cross interfering signal.}
	\label{fig:bp_kmmse_p}	
\end{figure}
~
\begin{figure}[tb]
	\centering
	\includegraphics[width=\linewidth]{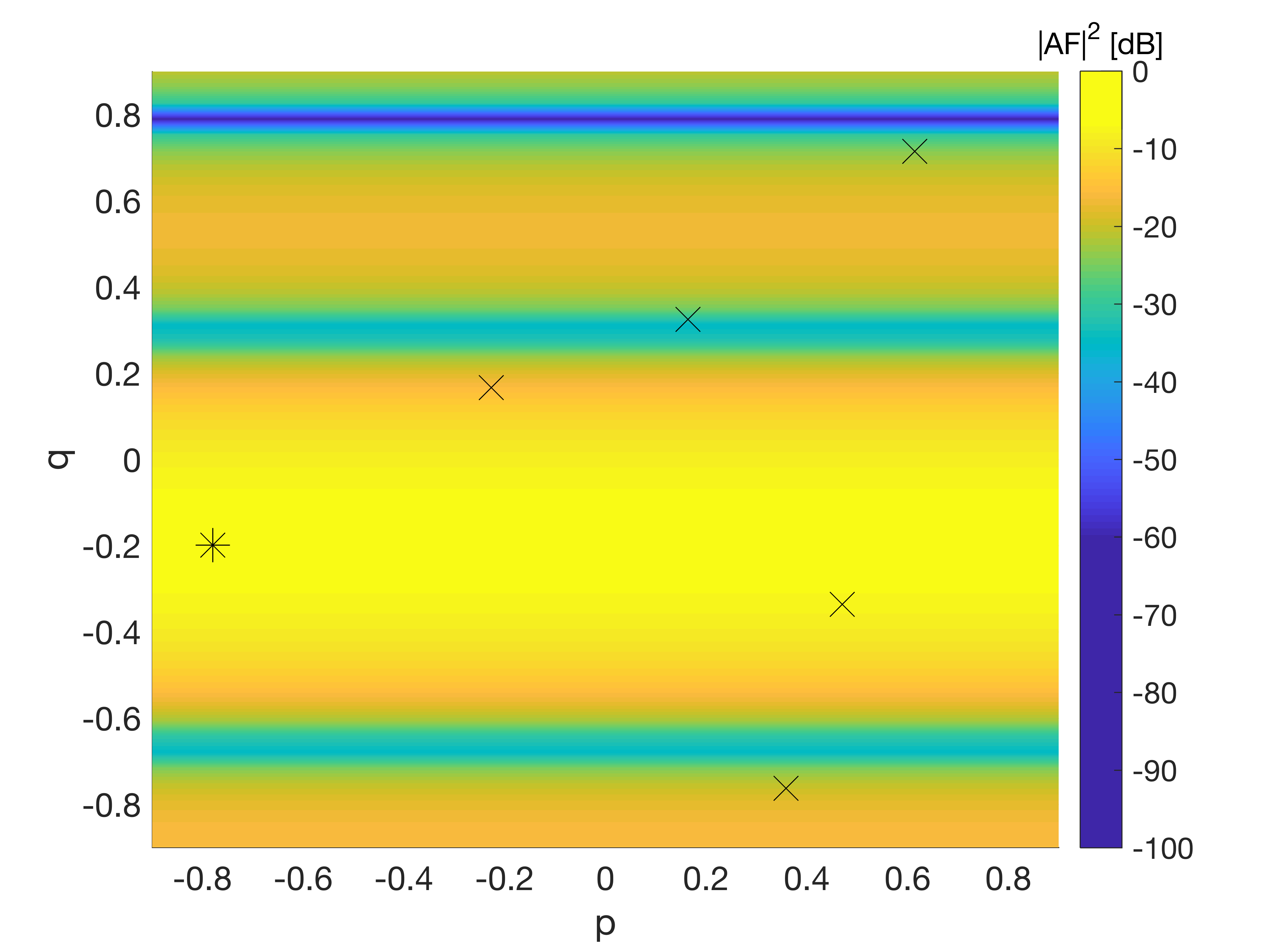}
	\caption{Vertical KMMSE AF squared magnitude. $R=6$ wavefronts, $N_v =4$, $\rho=0.5$. Asterisk denotes desired signal, cross interfering signal.}	
	\label{fig:bp_kmmse_q}		
\end{figure}

The computational complexity, measured in flops, is plotted as a function of the array size in Figure~\ref{fig:flops}. We use the MATLAB Lightspeed toolbox~\cite{minka_lightspeed_2007} for flops counting since it provides the approximate number of operations required for inverting matrices. This result shows that the proposed method substantially reduces the computational complexity of beamforming design. For an array of $16\times 16$ antennas, the complexity difference between MMSE and KMMSE is around three orders of magnitude. While KMMSE is inexpensive in all scenarios, TMMSE is costly for relatively small arrays due to the iterative optimisation procedure. For arrays of $8\times 8$ antennas, a set-up expected for 5G systems~\cite{ji_overview_2017}, both separable beamformers are less expensive than MMSE.

We investigate the influence of the regularization parameter $\rho$ on the KMMSE BER performance in Figure~\ref{fig:ber_lambda}. We observe that regularization plays a little role in the performance for low SNR ($< 0\,\si{\decibel}$). In this case, the noise term on \eqref{eq:subcov} has enough energy to complete the rank of the covariance matrix, decreasing its condition number~\cite{golub2012matrix}, thus making the horizontal and vertical MMSE beamformers numerically stable. However, for high SNR ($\geq 0\,\si{\decibel}$), regularization is paramount to achieve a satisfactory performance. This is because the noise term is not strong anymore to fill the covariance matrix rank in this case, and then the horizontal and vertical covariance matrices in \eqref{eq:subcov} become ill-conditioned. To tackle this issue, the regularization term fills the covariance matrix rank, decreasing the condition number, and thus making the regularized MMSE beamformers stable at high SNR. The link between the covariance matrix condition number and the regularization parameter $\rho$ is clarified in Figure~\ref{fig:cond_lambda}. We note that the covariance matrix becomes more well-conditioned at high SNR as we increase $\rho$. More specifically, for $\rho = 0$, the BER worsens with the SNR, owing to the increase in the condition number of the horizontal and/or vertical covariance matrices. For $\rho=0.1$, KMMSE exhibits a bad performance from $0,\si{\decibel}$ to $16\,\si{\decibel}$ SNR. In this range, the regularisation factor is not large enough to avoid performance deterioration due to ill-condition of the covariance matrix. Finally, for $\rho > 0.1$, we observe that KMMSE is not much affected by the covariance matrix condition. However, the regularisation term yields a BER bias, which increases with $\rho$. Figure~\ref{fig:ber_lambda} shows that $\rho=0.5$ provides the best stability-performance trade-off, thus we choose this value for the next simulations.

Unfortunately, the important computational complexity reduction observed in Figure~\ref{fig:flops} comes with source recovery degradation, as one can see in Figure~\ref{fig:ber_random}. This figure shows that TMMSE exhibits good performance from $-20\,\si{\decibel}$ to $0\,\si{\decibel}$, while KMMSE with regularization parameter $\rho=0.5$ performs similarly to the benchmark only from $-20\,\si{\decibel}$ to $-8\,\si{\decibel}$. At high SNR, the BER performance of the separable beamformers is heavily penalized compared to the benchmark. When one of the direction cosines ($p_r$ or $q_r$) of the interfering wavefronts is adjacent to those of the desired signal, the separable beamformers fail to recover it, yielding a significant number of bit errors. Since the direction cosines of all wavefronts are randomly selected according to a uniform distribution, it is rather common that the interfering wavefronts are close to the desired signal in either $p$- or $q$-domain. As a consequence, beamforming fails due to ill-conditioned covariance matrices, as discussed in the previous paragraph. At high SNR, the poor BER performance is especially accentuated because the noise component does not have enough power to fill the covariance matrix rank. However, whenever the wavefronts are sufficiently separated in space, the separable beamformers exhibit good source recovery performance. The presented results suggest that the proposed methods are appealing alternatives to the standard MMSE beamformer. In Figure~\ref{fig:flops}, the computational complexity of KMMSE, for example, is two orders of magnitude smaller than that of MMSE for $N_h \times N_v = 8\times 8$. On the other hand, Figure~\ref{fig:ber_random} indicates that KMMSE is $5\,\si{\decibel}$ apart from MMSE for the uncoded BER of $10^{-3}$.

To better understand why the separable beamformers are more sensitive to closely-spaced wavefronts, let us investigate their normalised array factor. We consider a scenario in which the proposed beamformers fail due to lack of degrees of freedom. For visualisation easiness, we consider a scenario where a $4\times 4$ array is applied to filter $R=6$ wavefronts in the following. In this case, each sub-array beamformer has only $4$ degrees of freedom and will fail to filter the $R=6$ signals. By contrast, the classical MMSE beamformer has $16$ degrees of freedom and is able to null the interfering wavefronts and recover the desired signal. Figures~\ref{fig:bp_mmse}, \ref{fig:bp_tmmse}, and \ref{fig:bp_kmmse}  show the magnitude of the MMSE, TMMSE, and KMMSE normalised array factors as functions of the direction cosines, respectively. One can see that the MMSE filter accurately places nulls at the interfering wavefronts directions, while a strong beam is pointed towards the desired signal. This is possible because the $16$-dimensional filter has sufficient degrees of freedom to separate the wavefronts. In contrast to the benchmark method, KMMSE does not accurately distribute nulls, hindering interference attenuation. In Figures~\ref{fig:bp_kmmse_p} and \ref{fig:bp_kmmse_q}, one can see that the KMMSE sub-beamformers do not have enough degrees of freedom to attenuate the interfering wavefronts. As a consequence, the undesired signal at $(p,q) =(0.5,-0.3)$ is not properly attenuated, as seen in Figure~\ref{fig:bp_kmmse}. We observe that only $3$ nulls are placed to attenuate $5$ interfering wavefronts in Figure~\ref{fig:bp_kmmse_p}. The same is observed in Figure~\ref{fig:bp_kmmse_q}. To solve this issue, one would need to increase the number of antennas to, at least, $N_h \times N_v = 6\times 6$. Figure~\ref{fig:bp_tmmse} reveals that TMMSE is more accurate than KMMSE at null placement. This is because the null locations are optimised as the alternating algorithm iterates. This accuracy is important especially at high SNR, as one can see in Figure~\ref{fig:ber_random}. We conclude that the separable beamformers are more sensitive to the number of impinging signals and closely spaced wavefronts than the classical MMSE beamformer due to the reduced degrees of freedom of the sub-beamformers. 

\section{Conclusion} \label{sec:conc}

We presented two beamforming methods that exploit array separability to reduce the computational complexity of the classical MMSE beamformer. The TMMSE filter is based on tensor algebra and minimises the MMSE by means of alternating minimisation, while the KMMSE filter relies on regularized sub-array MMSE beamforming. Our simulation results show that TMMSE provides moderate computational complexity reduction with small source recovery degradation. By contrast, KMMSE is computationally inexpensive but exhibits poorer BER performance at high SNR.  Therefore, TMMSE should be employed when source recovery performance is more important than computation efficiency, and KMMSE on the contrary. This work paves the way to future contributions, including the extension to massive MIMO architectures, e.g., hybrid analogue/digital transceivers. Furthermore, it would be of interest to validate the proposed methods with array responses simulated in an electromagnetic field simulator software.

\section{Appendix} \label{app:proof_urra}

It is well-known that the minimisers of \eqref{eq:it_out_lin1} and \eqref{eq:it_out_lin2} are given by the classical MMSE filters $\bm{w}_h = \bm{R}_{hh}^{-1} \bm{p}_{hs}$ and $\bm{w}_v = \bm{R}_{vv}^{-1} \bm{p}_{vs}$, respectively. In this appendix, we obtain the covariance matrices and cross-covariance vectors necessary to calculate these filters.  In our demonstrations, we consider only the horizontal sub-array. The statistics for the vertical sub-array are analogously derived.

First, let us represent \eqref{eq:tensor_model} in terms of the matrix unfoldings of $\mc{A}$. Unfolding this tensor along its first mode gives \cite{kolda_tensor_2009}
\begin{equation}
	\bm{X}[k] = [\mc{A}]_{(1)} (\bm{s}[k] \otimes \bm{I}_{N_v}) + \bm{B}[k].
\end{equation}
Now the horizontal sub-array input \eqref{eq:it_subsig_h} can be expressed as
\begin{align}
	\bm{u}_h[k] &\stackrel{(a)}{=} \left[[\mc{A}]_{(1)} (\bm{s}[k] \otimes \bm{I}_{N_v}) + \bm{B}[k]\right](1 \otimes \bm{w}_v^*)\\
	                    &\stackrel{(b)}{=} [\mc{A}]_{(1)} (\bm{s}[k] \otimes \bm{w}_v^*) + \bm{B}[k]\bm{w}_v^*, \label{eq:hinput}
\end{align}
where $(a)$ follows by considering $\bm{w}_v^* = 1 \otimes \bm{w}_v^*$,  and $(b)$ is the application of the mixed product property~\cite{liu_hadamard_2008}:
\begin{equation}\label{eq:kronprod}
(\bm{A} \otimes \bm{B})(\bm{C} \otimes \bm{D}) = (\bm{AC}) \otimes (\bm{BD})
\end{equation}
for any matrices $\bm{A}$, $\bm{B}$, $\bm{C}$, $\bm{D}$ with matching dimensions. The covariance matrix of \eqref{eq:hinput} is then given by
\begin{align}
	\bm{R}_{hh} = [\mc{A}]_{(1)} ( \bm{R}_{ss} \otimes \bm{w}_v^* \bm{w}_v^\tran ) [\mc{A}]_{(1)}^\hermit + \bm{R}_{cc}, \label{eq:cov_rh}
\end{align}
where $\bm{R}_{cc} = \esp{ \bm{c}[k] \bm{c}^\hermit[k]  }$, and $\bm{c}[k] = \bm{B}[k]\bm{w}_v^* \in \bmm{C}^{N_h}$. Note we have used the fact that $\esp{\bm{A} \otimes \bm{B}} = \esp{\bm{A}} \otimes \esp{\bm{B}}$ for matrices $\bm{A}$ and $\bm{B}$ with mutually independent elements, and that the inputs of $\bm{B}[k]$ and $\bm{s}[k]$ are uncorrelated. To calculate $\bm{R}_{cc}$, consider the element-wise representation of $\bm{c}[k]$:
\begin{equation}
[\bm{c}[k]]_{n_h} = \sum_{n_v=1}^{N_v} [\bm{B}[k]]_{n_h,n_v} [\bm{w}_v]_{n_v}^*,\quad n_h \in \{1,\ldots,N_h\}.
\end{equation}
The elements of $\bm{R}_{cc}$ are given by:
\begin{align}
&[\bm{R}_{cc}]_{n_h, n_h^\prime} = \esp{ [\bm{c}[k]]_{n_h} [\bm{c}[k]^\hermit]_{n_h^\prime} } \\
&= \esp{ \sum_{n_v=1}^{N_v} \sum_{n_v^\prime=1}^{N_v} \left[ \bm{B}[k] \right]_{n_h, n_v} \left[ \bm{B}[k] \right]_{n_h^\prime, n_v^\prime}^* \left[ \bm{w}_v \right]_{n_v} \left[ \bm{w}_v \right]_{n_v^\prime}^*  } \label{eq:noise_devel}
\end{align}
for $n_h,\,n_h^\prime \in \{1,\ldots, N_h \}$. As the AWGN vector has mutually independent elements, it follows that
\begin{equation}
\esp{\left[\bm{B}[k] \right]_{n_h, n_v} \left[ \bm{B}[k] \right]_{n_h^\prime, n_v^\prime}^*} = 0
\end{equation} 
for all $n_v \neq n_v^\prime$ and $n_h \neq n_h^\prime$. Therefore, the off-diagonal elements of $\bm{R}_{cc}$ are zero and those at the main diagonal are given by
	\begin{align}
[\bm{R}_{cc}]_{n_h, n_h} &= \sum_{n_v=1}^{N_v} \esp{\left[ \bm{B}[k] \right]_{n_h, n_v} \left[ \bm{B}[k] \right]_{n_h, n_v}^*} \left[\bm{w}_v\right]_{n_v} \left[\bm{w}_v\right]_{n_v}^* \\
&+\sigma_b^2 \sum_{n_v=1}^{N_v} \left[\bm{w}_v\right]_{n_v} \left[\bm{w}_v\right]_{n_v}^* = \sigma_b^2  \|\bm{w}_v\|^2_2 \label{eq:noise_devel_final}
\end{align}
and, consequently, we get $\bm{R}_{cc} = \sigma_b^2  \|\bm{w}_v\|^2_2 \bm{I}_{N_h}$, concluding the derivation of $\bm{R}_{hh}$. From the definition of $\bm{p}_{hs}$ and $\bm{u}_h[k]$, it follows that:
\begin{align}
\bm{p}_{hs} &= \esp{ \left( [\mc{A}]_{(1)} (\bm{s}[k] \otimes \bm{w}_v^*) + \bm{B}[k]\bm{w}_v^* \right) (s_d^*[k] \otimes 1)}\\
			   &= [\mc{A}]_{(1)} (\bm{s}[k] s_d^*[k] \otimes \bm{w}_v^*) = [\mc{A}]_{(1)} ( \Rss\edes \otimes \bm{w}_v^*),
\end{align}
finalising our proof. \qed

\section*{Acknowledgment}
This work is supported by the National Council for Scientific and Technological Development -- CNPq, CAPES/PROBRAL Proc. no. 88887.144009/2017-00, and FUNCAP.

\bibliographystyle{IEEEtran}
\bibliography{iet18}

\end{document}